\documentclass{llncs}
\usepackage{makeidx}  
\usepackage{amssymb}
\usepackage{textcomp}
\newcommand{\cent}[0]{\mbox{\textcent}}

\usepackage{graphicx}

\begin{document}

%
%
\pagestyle{headings}  

\title{From Quantum Query Complexity to \\State Complexity\thanks{Work of the first author was supported by
the Employment of Newly Graduated Doctors of Science for Scientific Excellence project/grant (CZ.1.07./2.3.00/30.0009)  of Czech Republic. Work of second author   was supported by the National
Natural Science Foundation of China (Nos. 61272058, 61073054).}}
\titlerunning{From Quantum Query Complexity to State Complexity}  
%
\author{Shenggen Zheng\inst{1}
 \and Daowen Qiu\inst{2}}
\authorrunning{S. Zheng, D. Qiu }   
%
\tocauthor{Shenggen Zheng,  Daowen Qiu}
\institute{
Faculty of Informatics, Masaryk University, Brno, 602 00, Czech Republic\\
\email{zhengshenggen@gmail.com}
\and
Department of Computer Science, Sun Yat-sen University, \\Guangzhou 510006, China\\
\email{ issqdw@mail.sysu.edu.cn}
}

\maketitle              

\begin{abstract}
State complexity of quantum finite automata is one of the  interesting topics in studying the power of quantum finite automata. It is therefore of importance to develop general methods how to show state succinctness results for quantum finite automata. One such  method is presented and demonstrated in this paper. In particular, we show that state  succinctness results can be derived out of query complexity results.

\end{abstract}
\section{Introduction}

An important way to get deeper insights into the power of various quantum resources and operations is to explore the power of various quantum variations of the basic models of classical  automata. Of a special interest is to do that for various quantum variations of the classical  automata, especially for those models that use very limited amounts of quantum  resources: states, correlations, operations and measurements. This paper aims to contribute to such a line of research.

Number of (basis) states used  is a natural complexity measure for (quantum) finite automata. The size of a (quantum) finite automaton is defined as the number of (basis) states of the (Hilbert) space on which the automaton will operate.
 In case of a hybrid, that is quantum/classical finite automata, it is natural to consider both complexity measures~--~the number of classical states and also the number of quantum (basis) states.

Quantum finite automata were  introduced by  Kondacs and Watrous \cite{Kon97} and also by Moore and Crutchfields \cite{Moo97}, and since that time they were intensively explored  \cite{Amb02,Bro99,Qiu11,Yak11}.  State complexity and succinctness results are an important
research area of the classical finite automata theory, see \cite{Yu95}, with a
variety of applications. Once quantum versions of classical finite automata were introduced  and explored, it started to be of large interest to find out, also through succinctness results, a relation between the power of classical and
 quantum finite automata models. This has turned out to be an area of surprising
 outcomes that again indicated that the relations between
 the classical and the corresponding quantum finite automata models are intriguing.
  In the past twenty years, state complexity of several variants of quantum finite automata were deeply and broadly studied \cite{Amb98,AmbNay02,Amb09,AmYa11,Ber05,BMP06,Fre09,Gru00,GQZ14,GQZ14b,Le06,Mer01,Mer02,Qiu14,Yak10,ZhgQiu112,Zhg12,Zhg13}.

State succinctness results were proved for some  special languages and  promise problems and for several automata models. The methods used to prove those results are various and often ad hoc.
It is therefore natural  to try to find out whether there are quite general methods to get state succinctness results for quantum finite automata.
 The answer is yes. We will show, in this paper, that state succinctness results can be derived in a nice way  out of  query complexity results. Here is the basic idea: State complexity is deeply related to communication complexity \cite{KusNis97b}.
  Buhrman et al. proved that various communication complexity results can be derived out of query complexity results \cite{Buh98}. If a communication protocol is simple enough, then we can use  quantum finite automata to implement it. By using this line of thought, state succinctness results can be derived.

 Quantum query complexity is the quantum generalization of the model of  decision tree complexity. In this model, an algorithm to compute a Boolean function $f:\{0,1\}^n\to \{0,1\}$ is charged for ``queries" to the input bits, while any intermediate computation is considered as free (see \cite{BdW02}).

Communication complexity was introduced by Yao \cite{Yao79} in 1979.
In the setting of two  parties, Alice is given   $x\in\{0,1\}^n$, Bob is given $y\in\{0,1\}^n$ and their task is to communicate in order to determine the value of some Boolean function $f:\{0,1\}^n\times\{0,1\}^n\to\{0,1\}$, while exchanging as small  number of bits   as possible. In this model, local computation
is considered to be free, but communication is considered to be  expensive and has to be minimized.
Moreover, for computation, Alice and Bob  can use all   power available.
There are usually three types of communication complexities considered according to the models of protocols used by Alice and Bob: deterministic, probabilistic and quantum.

Query complexity and communication complexity are related to each other. By using a simulation technique that transforms quantum query algorithms
to quantum communication protocols,  Buhrman et al. \cite{Buh98,Buh09} obtained new quantum communication protocols and showed  the first impressively  (exponential) gap between quantum and classical communication complexity. In the reverse direction, Buhrman et al. showed that how  to use lower
bounds for quantum communication protocols to derive lower bounds
for quantum query algorithms.

State complexity of finite automata and  communication complexity  are also related to each other.
We can use communication complexity results to prove  lower bounds on state complexity  \cite{Hro01,Kla00,KusNis97b}. On the other hand, if the communication protocols are easy enough, then they can be simulated  by finite automata and obtain new state complexity results (upper bounds) for finite automata.

Therefore, we can build connections from query complexity to state complexity. This could be a potential framework to get state succinctness results for quantum finite automata comparing to classical finite automata. We will demonstrate for several cases in this paper, that how to use   quantum query complexity results  to  derive state  succinctness results of finite automata.

We  first consider  the promise problem (partial function) studied in \cite{MJM11}. Namely, the problem
\begin{equation}
\mbox{DJ}'(x)=
\left\{
  \begin{array}{ll}
    1 & \hbox{if } W(x)\in\{0, 1, n-1, n\}\\
    0 & \hbox{if }W(x)=\frac{n}{2},
  \end{array}
\right.
\end{equation}
where $W(x)$ is the Hamming weight of $x$.
Montanaro et al. \cite{MJM11}  gave a quantum query algorithm for  $\mbox{DJ}'$ with 2 queries. However, their proof is quite complicated. Motivated by the method from \cite{AISJ13}, we  give a simpler quantum query algorithm with 2 queries for $\mbox{DJ}'$.

 Based on this simple query algorithm, we design a quantum communication protocol for the following promise problem

\begin{equation}
\mbox{EQ}'(x,y)=\left\{\begin{array}{ll}
                    1 &\ \mbox{if}\ H(x,y)\in\{0, 1, n-1, n \} \\
                    0 &\ \mbox{if}\ H(x,y)=\frac{n}{2},
                  \end{array}
 \right.
\end{equation}
where  $H(x,y)$ is the Hamming distance between bit strings $x$ and $y$. We further prove that the exact quantum communication complexity of $\mbox{EQ}'$ is ${\bf O}(\log n)$ while the deterministic communication complexity  is ${\bf \Omega}( n)$.

Finally, we consider the  promise problem $A(n)=(A_{yes}(n),A_{no}(n))$, where $A_{yes}(n)=\{x\#y\#\#x\#y\,|\,H(x,y)\in\{0, 1, n-1, n \}, \ x,y\in\{0,1\}^n\}$ and $ A_{no}(n)=\{x\#y\#\#x\#y\,|\,H(x,y)=\frac{n}{2},\ x,y\in\{0,1\}^n \}.$
We will prove that the promise problem $A(n)$ can be solved exactly by  a one-way finite automata with quantum and classical state (1QCFA)  with ${\bf O}(n^2)$ quantum basis states and ${\bf O}(n^3)$ classical states, whereas the sizes of the corresponding one-way deterministic finite automata (1DFA)  are $2^{{\bf \Omega}(n)}$.

The paper is structured as follows. In Section 2  basic concepts and notations are introduced and  models involved are
described in some details. A new quantum query algorithm is given for $\mbox{DJ}'$ in Section~3. Communication complexity of $\mbox{EQ}'$  is explored in Section~4. State complexity results for the promise problem $A(n)$ are showed in Section~5.
\section{Preliminaries}

In this section, we recall some  basic definitions about query complexity, communication complexity and quantum finite automata. Concerning basic concepts and notations of quantum information processing and  finite automata, we refer the reader to  \cite{Gru97,Gru99,Gru00,Qiu12}.

\subsection{Exact query complexity}

 Exact quantum query complexity for partial functions  was dealt with in \cite{BH97,DJ92,GQZ14} and for total functions  in \cite{Amb13,AISJ13,AGZ14,BBC+98,MJM11}.
 Concerning more basic concepts and notations concerning query complexity, we refer the reader to \cite{BdW02}.

An exact classical (deterministic) query algorithm for computing a Boolean function $f:\{0,1\}^n\to \{0,1\}$ can be described
by a decision tree. A decision tree $T$ is a rooted binary tree where each internal vertex
has exactly two children, each internal vertex is labeled with a variable $x_i$ and each leaf is labeled with a value 0 or 1. $T$ computes a Boolean function $f$ as follows: The start is at the root. If this is a leaf then stop. Otherwise, query the variable $x_i$ that labels the root. If $x_i=0$, then recursively evaluate the left subtree, if $x_i=1$ then recursively evaluate the right subtree. The output of the tree is the value of the leaf that is reached at the end of this
process.
The depth
of $T$ is the maximal length of a path from the root to a leaf (i.e. the worst-case number of queries
used on any input). The {\em exact classical query complexity} (deterministic
query complexity, decision tree complexity)  is the minimal depth over all decision trees computing $f$.

Let $f:\{0,1\}^n\to \{0,1\}$ be a Boolean function and $x = x_1x_2\cdots x_n$ be an input
bit string.   An exact quantum query algorithm for $f$
works in a Hilbert space with some fixed number of basis states. It starts in a
fixed starting state, then performs on it a sequence of  transformations
$U_1$, $Q$, $U_2$, $Q$, \ldots, $U_t$, $Q$, $U_{t+1}$.
Unitary transformations $U_i$ do not depend on
the input bits, while $Q$, called the  query transformation, does,
in the following way. Each of the basis states corresponds to either one or none
of the input bits. If the basis state $|\psi\rangle$ corresponds to the $i$-th
input bit, then $Q|\psi\rangle=(-1)^{x_i}|\psi\rangle$. If it does not correspond to any
input bit, then $Q$ leaves it unchanged: $Q|\psi\rangle=|\psi\rangle$. Finally, the algorithm performs a  measurement in the standard basis.
Depending on the result of the measurement, the algorithm outputs either 0 or 1
which must be equal to $f(x)$. The {\em exact quantum query complexity}  is the minimum number of queries made by any quantum algorithm computing $f$.

\subsection{Communication complexity}

We recall here only  very  basic concepts and notations of  communication complexity,
and we refer the reader to \cite{KusNis97b} for
more details. We will deal with the situation that there are  two
communicating  parties and with very simple tasks of
computing two inputs Boolean functions for the case  one input is known to
one party and the other input to the other party.
We will completely ignore computational resources needed by
parties and  focus solely on the amount of communication  need to be
exchanged between both parties in order to compute the value of a given Boolean function.

More technically, let $X, Y$ be finite subsets of $\{0,1\}^n$. We will consider  two-input functions $f: X\times Y\rightarrow \{0,1\}$ and two communicating parties.  Alice is given  $x\in X$ and Bob is given  $y\in Y$. They want to compute $f(x,y)$. If $f$ is  defined only on a proper subset of $X\times Y$,  $f$ is said to be a partial function or a promise problem.

The computation of   $f(x,y)$ will be  done using a communication protocol. During the
execution of the protocol, the two parties alternate roles in sending messages. Each of these
messages is a bit-string. The protocol, whose steps are based on the communication so far, specifies also for each step whether the communication terminates
(in which case it also specifies what is the output). If the communication is
not to terminate, the protocol specifies what kind of  message the sender (Alice or Bob) should send next, as
a function of its input and communication so far.

A deterministic communication protocol ${\cal P}$ computes
a (partial) function $f$, if for every (promised) input pair $(x,y)\in X\times Y$ the protocol terminates with the
value $f(x,y)$ as its output.
In a probabilistic  protocol, Alice and Bob may also flip coins during the protocol execution and proceed according to its output and the protocol can also have an erroneous output with a small probability.
In a  quantum protocol, Alice and Bob may   use quantum resources to produce the output or (qu)bits  for communication.

 Let ${\cal P}(x,y)$ denote the  output of the protocol ${\cal P}$.
 For an exact protocol, that always outputs the correct answer, $Pr({\cal P}(x,y)=f(x,y))=1$.

The communication complexity of a protocol ${\cal P}$  is the
worst case number of (qu)bits exchanged.  The communication complexity of $f$ is, with  which respect to the communication mode used,  the complexity of an
optimal protocol for $f$.

We will use $D(f)$ to denote the {\em deterministic communication complexity} and  $Q_E(f)$ to denote  the {\em exact   quantum communication complexity}.

\subsection{One-way finite automata with quantum and classical states}

In this subsection we recall the definition of 1QCFA.

{\em Two-way finite automata with quantum and classical states} were introduced by Ambainis and Watrous \cite{Amb02} and then explored in \cite{Yak10,ZhgQiu11,ZhgQiu112,Zhg12,Zhg13,Zhg13a}.
 1QCFA are one-way versions of 2QCFA,   which  were introduced by   Zheng et al. \cite{ZhgQiu112}.
Informally, a 1QCFA can be seen as a 1DFA which has access to a quantum
 memory of a constant size (dimension), upon which it performs
quantum transformations and measurements. Given a finite set of quantum basis states $Q$, we denote by $\mathcal{H}(Q)$
the Hilbert space spanned by $Q$. Let
$\mathcal{U}(\mathcal{H}(Q))$ and $\mathcal{O}(\mathcal{H}(Q))$
denote the sets of unitary operators and projective measurements over $\mathcal{H}(Q)$, respectively.

\begin{definition}
A {\it one-way finite automaton with quantum and classical states} $\mathcal{A}$ is specified by a 10-tuple
\begin{equation}
\mathcal{A}=(Q,S,\Sigma,\Theta,\Delta,\delta,|q_{0}\rangle,s_{0},S_{acc},S_{rej})
\end{equation}
where:
\begin{enumerate}
\item $Q$ is a finite set of orthonormal quantum basis states.
\item $S$ is a finite set of classical states.

\item $\Sigma$ is a finite alphabet of input symbols and let
$\Sigma'=\Sigma\cup \{\cent ,\$\}$, where $\cent$ will be used as the left end-marker and $\$$ as the right end-marker.

\item $|q_0\rangle\in Q$ is the initial quantum state.
\item $s_0$ is the initial classical state.
\item $S_{acc}\subset S$ and $S_{rej}\subset S$, where $S_{acc}\cap S_{rej}=\emptyset$, are  sets of
the classical accepting and rejecting states, respectively.

\item $\Theta$ is a quantum transition function
\begin{equation}
\Theta: S\setminus(S_{acc}\cup S_{rej})\times \Sigma'\to \mathcal{U}(\mathcal{H}(Q)),
\end{equation}
assigning to each pair $(s,\gamma)$ a
 unitary transformation.

\item $\Delta$ is a mapping
\begin{equation}
\Delta:S\times \Sigma'\rightarrow
\mathcal{O}(\mathcal{H}(Q)),
\end{equation}
where each $\Delta(s,\gamma)$ corresponds to a projective measurement
{\em (}a projective measurement will be taken each time a unitary
   transformation is applied; if we do not need a measurement,
we denote that $\Delta(s,\gamma)=I$, and we assume the result of the measurement to be a fixed $c${\em )}.

\item $\delta$ is a special transition function of classical states.
Let the results set of the measurement be $\mathcal{C}=\{c_{1},c_{2},\dots$,
$c_{s}\}$, then
\begin{equation}
\delta:S\times \Sigma' \times \mathcal{C}\rightarrow
S,
\end{equation}
 where $\delta(s,\gamma)(c_{i})=s'$ means
that if a tape symbol $\gamma\in \Sigma'$ is being
 scanned and the projective measurement result is $c_{i}$, then the  state $s$ is changed to $s'$.
\end{enumerate}
\end{definition}

Given an input $w=\sigma_1\cdots\sigma_l$, the word on the tape will be
$w=\cent w\$$ (for convenience, we denote $\sigma_0=\cent $ and $\sigma_{l+1}=\$$).
Now, we define the behavior of 1QCFA $\mathcal{A}$ on the input word $w$.
The computation starts in the classical state $s_0$ and the quantum state $|q_0\rangle$, then the
transformations associated with symbols in the word  $\sigma_0\sigma_1\cdots, \sigma_{l+1}$ are applied in succession.
The transformation associated
with a state $s\in S$ and a symbol $\sigma\in\Sigma'$ consists of three steps:
\begin{enumerate}
\item Firstly, $\Theta(s,\sigma)$ is applied to the current quantum state
$|\phi\rangle$, yielding the new state
$|\phi'\rangle=\Theta(s,\sigma)|\phi\rangle$.

\item Secondly, the observable $\Delta(s,\sigma)=\mathcal{O}$ is measured on
$|\phi'\rangle$. The set of possible results is $\mathcal{C}=\{c_1, \cdots, c_s\}$.
According to   quantum mechanics principles, such a
measurement yields the classical outcome $c_k$ with probability
$p_k=||P(c_k)|\phi'\rangle||^2$, and the quantum state of $\mathcal{A}$
collapses to $P(c_k)|\phi'\rangle/\sqrt{p_k}$.

\item Thirdly, the current classical state $s$ will be changed to $\delta(s,\sigma)(c_k) =s'.$

\end{enumerate}
An input word $w$ is assumed to be accepted (rejected) if and only if the classical state after scanning $\sigma_{l+1}$
is an accepting (rejecting) state. We assume that $\delta$ is well defined so that 1QCFA $\mathcal{A}$ always accepts or
rejects at the end of the computation.

Language acceptance is a special case of so called promise problem solving.
A {\em promise problem} is a pair $A = (A_{yes}, A_{no})$, where $A_{yes}$, $A_{no}\subset \Sigma^*$
are disjoint sets. Languages may be viewed as promise problems that obey the additional constraint
$A_{yes}\cup A_{no}=\Sigma^*$.

 A promise problem $A = (A_{yes}, A_{no})$ is solved  exactly by a finite automaton ${\cal A}$  if
\begin{enumerate}
\item[1.] $\forall w\in A_{yes}$, $Pr[{\cal A}\  \mbox{accepts}\  w]=1$, and
\item[2.] $\forall w\in  A_{no}$, $Pr[{\cal A}\ \mbox{rejects}\  w]=1$.
\end{enumerate}

\section{An exact quantum query algorithm for $\mbox{DJ}'(x)$}
Montanaro et al. \cite{MJM11}  gave a quantum algorithm for  $\mbox{DJ}'$ with 2 queries. However, their proof is  complicated.
Motivated by the method from \cite{AISJ13},
 we  give a simpler algorithm with 2 queries for $\mbox{DJ}'$ as follow:

We  use basis states $|0,0\rangle$, $|i,0\rangle$,  $|i,j\rangle$ and $|k\rangle$ with $0\leq i<j\leq n$ and $1\leq k\leq n-2$.
A basis state $|i,j\rangle$ corresponds to an input bit $x_i$ for $1\leq i\leq n$; a basis state $|k\rangle$ corresponds to an input bit $y_k$ for $1\leq k\leq n-2$ ($y_k$ is some certain bit $x_i$) and the other basis states do not correspond to any input bit.
\begin{enumerate}
  \item
The algorithm  ${\cal A}$ begins in the  state $|0,0\rangle$ and then a unitary mapping $U_1$ is applied on it:
\begin{equation}
    U_1|0,0\rangle=\sum_{i=1}^n\frac{1}{\sqrt{n}}|i,0\rangle.
\end{equation}
 \item
${\cal A}$ then  performs the query:
\begin{equation}
 \sum_{i=1}^n\frac{1}{\sqrt{n}}|i,0\rangle\to  \sum_{i=1}^n\frac{1}{\sqrt{n}}(-1)^{x_i}|i,0\rangle.
\end{equation}

 \item ${\cal A}$  performs a unitary mapping $U_2$ to the current state such that
 \begin{equation}
 U_2|i,0\rangle=\sum_{j>i\geq 1}\frac{1}{\sqrt{n}}|i,j\rangle-\sum_{1\leq j<i}\frac{1}{\sqrt{n}}|j,i\rangle+\frac{1}{\sqrt{n}}|0,0\rangle
 \end{equation}
 and the resulting quantum state will be
  \begin{equation}
  U_2\sum_{i=1}^n\frac{1}{\sqrt{n}}(-1)^{x_i}|i,0\rangle=\frac{1}{n}\sum_{i=1}^n(-1)^{x_i}|0,0\rangle+\frac{1}{n}\sum_{1\leq i<j}((-1)^{x_i}-(-1)^{x_j})|i,j\rangle.
   \end{equation}
   \item ${\cal A}$  measures the resulting state in the standard basis.  If the outcome is  $|0,0\rangle$, then $\sum_{i=1}^n(-1)^{x_i}\neq 0 $ and $\mbox{DJ}'(x)=1$. Otherwise, suppose that we get the state $|i,j\rangle$, then we have $x_i\neq x_j$. Let $y=x\setminus \{x_i,x_j\}$, we have $W(y)\in\{0,n-2,\frac{n-2}{2}\}$. If $W(y)=\frac{n-2}{2}$, then $\mbox{DJ}'(x)=0$. If  $W(y)\in \{0,n-2\} $, then $\mbox{DJ}'(x)=1$.
The remaining question is exactly the Deutsch-Jozsa promise problem \cite{DJ92} and we can get the answer with 1 query as follows:
we use the  subalgorithm ${\cal B}$ to solve the remaining promise problem using $n-2$  quantum basis  states $|1\rangle,\ldots,|n-2\rangle$ that will work as follows:
\begin{enumerate}
 \item ${\cal B}$  begins in the state $|1\rangle$ and performs on it a unitary transformation $U_3$ such that
   \begin{equation}
   U_3|1\rangle=\sum_{k=1}^{n-2} \frac{1}{\sqrt{n-2}}|k\rangle.
      \end{equation}
 \item ${\cal B}$  performs a query $Q$:
   \begin{equation}
   \sum_{k=1}^{n-2} \frac{1}{\sqrt{n-2}}|k\rangle\to   \sum_{k=1}^{n-2} \frac{1}{\sqrt{n-2}}(-1)^{y_k}|k\rangle
      \end{equation}
 \item ${\cal B}$  performs a  unitary transformation $U_4=U_3^{-1}$ and
  \begin{equation}
   U_3^{-1}  \sum_{k=1}^{n-2} \frac{1}{\sqrt{n-2}}(-1)^{y_k}|k\rangle=\frac{1}{n-2}\sum_{k=1}^{n-2}(-1)^{y_k}|1\rangle+\sum_{k=2}^{n-2}\beta_k|k\rangle,
      \end{equation}
      where $\beta_k$ are  amplitudes that we do not need to be specified  exactly.

  \item ${\cal B}$  measures the resulting state in the standard basis and outputs 1 if the measurement outcome is $|1\rangle$ and  0 otherwise.
 \end{enumerate}

 According to \cite{AISJ13}, the unitary mapping $U_2$ exists.  The rest of the proof is easy to verify. Obviously, the algorithm ${\cal A}$ uses 2 queries.
 \end{enumerate}

\section{Communication complexity of $\mbox{EQ}'(x,y)$}
In this section, we will prove that  $Q_E(\mbox{EQ}')$ is ${\bf O}(\log n)$ while $D(\mbox{EQ}')$ is ${\bf \Omega}(n)$.

\begin{theorem}\label{Th-1}
$Q_E(\mbox{EQ}')\in {\bf O}(\log n)$.
\end{theorem}
\begin{proof}
Assume that Alice is given an input $x=x_1,\cdots,x_n$ and Bob an input $y=y_1,\cdots,y_n$. The following quantum communication protocol ${\cal P}$ computes $\mbox{EQ}'$ using ${\bf O}(n^2)$ quantum basis  states  $|0,0\rangle$, $|i,0\rangle$,  $|i,j\rangle$ and $|k\rangle$ with $0\leq i<j\leq n$ and $1\leq k\leq n-2$
 as follows:

  \begin{enumerate}
  \item Alice begins with the quantum state $|\psi_0\rangle=|0,0\rangle$ and  performs on it the unitary map $U_1$. The  quantum state is changed to
        \begin{equation}
        |\psi_1\rangle=U_1|0,0\rangle=\frac{1}{\sqrt{n}}\sum_{i=1}^n|i,0\rangle.
            \end{equation}
  \item Alice applies the unitary map $U_x$ to the current state such that  $U_x|i,0\rangle=(-1)^{x_i}|i,0\rangle$ for $i>0$ and the quantum state is changed to
     \begin{equation}
     |\psi_2\rangle=U_x|\psi_1\rangle=\frac{1}{\sqrt{n}}\sum_{i=1}^n(-1)^{x_i}|i,0\rangle.
        \end{equation}

  \item Alice then sends her current quantum state $|\psi_2\rangle$ to Bob.

  \item Bob applies the unitary map $U_y$ to the state that he has received such that  $U_y|i,0\rangle=(-1)^{y_i}|i,0\rangle$ for $i>0$ and  the  quantum state is changed to
  \begin{equation}
  |\psi_3\rangle=U_y|\psi_2\rangle U_x=\frac{1}{\sqrt{n}}\sum_{i=1}^n(-1)^{x_i+y_i}|i,0\rangle.
   \end{equation}

  \item Bob applies the unitary map $U_2$ to his quantum state and the  quantum state is changed to
   \begin{equation}
   |\psi_4\rangle=U_2|\psi_3\rangle=\frac{1}{n}\sum_{i=1}^n(-1)^{x_i+y_i}|0,0\rangle+\frac{1}{n}\sum_{1\leq i<j}((-1)^{x_i+y_i}-(-1)^{x_j+y_j})|i,j\rangle.
   \end{equation}
 \item Bob  measures the resulting state in the standard basis and outputs 1 if the measurement outcome is $|0,0\rangle$. Otherwise, suppose that the outcome is $|i,j\rangle$. Bob sends $i$ and $j$ to Alice using classical bits.

  \item After Alice receives $i$ and $j$,  let $x'=x_1\ldots x_{i-1}x_{i+1}\ldots x_{j-1}x_{j+1}\ldots x_n$. (In convenience, we write   $x'=x'_1\ldots x'_{n-2}$).
   Alice applies $U_{3}$ to the basis state $|1\rangle$ such that the quantum state is changed to
    \begin{equation}
   |\psi_5\rangle=U_3|1\rangle=\frac{1}{\sqrt{n-2}}\sum_{k=1}^{n-2} |k\rangle.
     \end{equation}

 \item Alice then applies $U_{x'}$ to the current state such that   $U_{x'}|k\rangle=(-1)^{x'_k}|k\rangle$ for $k>0$ and the quantum state is changed to
 \begin{equation}
 |\psi_6\rangle= \frac{1}{\sqrt{n-2}}\sum_{k=1}^{n-2}(-1)^{x_k'}|k\rangle.
   \end{equation}

  \item Alice sends her current quantum state $|\psi_6\rangle$ to Bob.

  \item Bob applies the unitary map $U_{y'}$ to the state that he has received such that   $U_{y'}|k\rangle=(-1)^{y'_k}|k\rangle$ for $k>0$, where $y'=y_1\ldots y_{i-1}y_{i+1}\ldots y_{j-1}y_{j+1}\ldots y_n$. (In convenience, we write   $y'=y'_1\ldots y'_{n-2}$). The quantum state is changed to
       \begin{equation}
      |\psi_7\rangle=\frac{1}{\sqrt{n-2}}\sum_{k=1}^{n-2} (-1)^{x_k'+y_k'}|k\rangle.
    \end{equation}

  \item Bob performs a  unitary transformation $U_4=U_3^{-1}$ to the current state and the quantum state is changed to
    \begin{equation}
      |\psi_8\rangle=U_4 |\psi_7\rangle=\frac{1}{n-2}\sum_{k=1}^{n-2} (-1)^{x_k'+y_k'}|1\rangle+\sum_{k=2}^{n-2}\beta_k|k\rangle,
    \end{equation}
        where $\beta_k$ are  amplitudes that we do not need to be specified  exactly.

    \item Bob  measures the resulting state in standard basis and outputs 1 if the measurement outcome is $|1\rangle$ and outputs 0 otherwise.
\end{enumerate}
 The unitary transformations $U_1$, $U_2$, $U_3$ and $U_4$ are the same ones as defined in Section~3.

If $H(x,y)\in\{0,n\}$, then the quantum state in Step~5  is
  \begin{equation}
|\psi_4\rangle=\frac{1}{n}\sum_{i=1}^n(-1)^{x_i+y_i}|0,0\rangle=\pm|0,0\rangle.
  \end{equation}
 Bob will get the quantum state $|0,0\rangle$  after the measurement in Step~6 and output 1 as the result of $\mbox{EQ}'(x,y)$.

If $H(x,y)\in\{1,n-1\}$, then there are two cases:
\begin{description}
  \item[a)] If the  measurement outcome in Step~6 is  $|0,0\rangle$ and Bob outputs 1 as the result of $\mbox{EQ}'(x,y)$.
  \item[b)] If the  measurement outcome in Step~6 is  $|i,j\rangle$, then $H(x', y')\in\{0,n-2\}$ and the quantum state in Step~11 is
 \begin{equation}
      |\psi_8\rangle=\frac{1}{n-2}\sum_{k=1}^{n-2} (-1)^{x_k'+y_k'}|1\rangle=\pm |1\rangle
    \end{equation}
 Bob will get the quantum state $|1\rangle$  after the measurement in Step~12 and output 1 as the result of $\mbox{EQ}'(x,y)$.
\end{description}

If $H(x,y)=\frac{n}{2}$, then Bob will output 0 as the result of $\mbox{EQ}'(x,y)$ in Step~12.

In Step~3 Alice sends $\lceil\log (n^2)\rceil$ qubits, in Step~6 Bob sends  $2\lceil\log (n)\rceil$ bits and in Step~9 Alice sends  $\lceil\log (n-2)\rceil$  qubits. Since we can use qubits to send  bits, it is clear that this protocol  uses only ${\bf O}(\log n)$ qubits for communication.
\end{proof}

The proof for deterministic communication lower bound is similar to  the ones  in \cite{BdW02,Buh09}.  In order to obtain an exponential quantum speed-up in \cite{Buh09},  $\frac{n}{2}$ must be an even integer in the distributed Deutsch-Jozsa promise problem (see \cite{GQZ14} for argument). However, $\frac{n}{2}$ can be arbitrary integer in the promise problem $\mbox{EQ}'$ in this paper.

We use so called ``rectangles" lower bound method \cite{KusNis97b} to prove the result.

A {\em rectangle} in $X\times Y$ is a subset $R\subseteq X\times Y$ such that $R=A\times B$ for some $A\subseteq X$ and $ B\subseteq Y$. A rectangle $R=A\times B$ is called
$1(0)$-rectangle of a function $f:X\times Y\to \{0,1\}$
if for every $x\in A$ and $y\in B$ the value of $f(x,y)$ is 1 (0). Moreover, $C^i(f)$ is defined  as the minimum number of $i$-rectangles that partition the space of $i$-inputs (such inputs $x$ and $y$ that $f(x,y)=i$) of $f$.

\begin{lemma}\label{lm-d-lowbound} {\em \cite{KusNis97b}}
$D(f)\geq \max\{\log{ C^1(f)},\log{ C^0(f)}\}$.
\end{lemma}

\begin{remark}
 For  a partial
function $f: X\times Y\to \{0,1\}$ with domain $\mathcal{ D}$, a rectangle $R=A\times B$ is called
$1(0)$-rectangle if  the value of $f(x,y)$ is 1(0) for every $(x,y)\in \mathcal{ D}\cap (A\times B)$  -- we do not care about  values for $(x,y)\not\in \mathcal{ D}$.
The above lemma still holds for promise problems (that is for partial functions).
\end{remark}

\begin{theorem}\label{Th-lowerbound}
{\em $D(\mbox{EQ}')\in {\bf \Omega}( n)$}.
\end{theorem}
\begin{proof}
Let  ${\cal P}$ be a deterministic protocol for $\mbox{EQ}'$.  There are two cases:

{\bf Case 1:} $\frac{n}{2}$ is even. We  consider the set $E=\{(x,x),(x,\overline{x}) \,|\, W(x)=\frac{n}{2}\}$. For every $(x,y)\in E$, we have ${\cal P}(x,y)=1$. Suppose there is a 1-monochromatic rectangle $R=A\times B\subseteq \{0,1\}^n\times\{0,1\}^n$ such that ${\cal P}(x,y)=1$ for every promise pair $(x,y)\in R$. Let $S=R\cap E$. For $x,y\in\{0,1\}^n$, let us denote $|x\wedge y|=\sum_{i=1}^n x_i\wedge y_i$.
We now prove that for any distinct $(x,x'),(y,y')\in S$, $|x\wedge y|\neq \frac{n}{4}$.

According to the assumption that $(y,y')\in S\subset E$, we have $y'=y$ or $y'=\overline{y}$.  If $|x\wedge y|=\frac{n}{4}$, then $H(x,y)=2(\frac{n}{2}-\frac{n}{4})=\frac{n}{2}=H(x,\overline{y})$ and ${\cal P}(x,y')=0$. Since $(x,x')\in R$ and $(y,y')\in R$, we have $(x,y')\in R$ and  ${\cal P}(x,y')=0$, which is a contradiction.

 According to Corollary 1.2 from \cite{Fr87}, we have $|S|\leq 1.99^n$. Therefore, the minimum number of 1-monochromatic rectangles that partition the space of inputs is
 \begin{equation}
     C^1(\mbox{EQ}')\geq\frac{|E|}{|S|}\geq \frac{2{n\choose  n/2}}{(1.99)^n}>\frac{2^{n+1}/n}{(1.99)^n}.
\end{equation}
The deterministic communication complexity is then
 \begin{equation}
D(\mbox{EQ}')\geq \log{C^1(\mbox{EQ}')}>\log{\frac{2^{n+1}/n}{(1.99)^n}}>0.0073n.
 \end{equation}

 {\bf Case 2:} $\frac{n}{2}$ is odd. We assume that $n=4k+2$. We  consider the set $E=\{(x,x') \,|\, W(x)=\frac{n}{2}\}$, where $x_n'=x_n=1$ and $x_i'=1-x_i$ for $i<n$.  For every $(x,x')\in E$, we have $H(x,x')=n-1$ and ${\cal P}(x,x')=1$. Suppose there is a 1-monochromatic rectangle $R=A\times B\subseteq \{0,1\}^n\times\{0,1\}^n$ such that ${\cal P}(x,y)=1$ for every promise pair $(x,y)\in R$. Let $S=R\cap E$. We now prove that for any distinct $(x,x'),(y,y')\in S$, $\sum_{i=1}^{n-1}|x_i\wedge y_i|\neq k$, that is $|x\wedge y|\neq k+1$.

 If $|x\wedge y|= k+1$,  without a lost of generality, let $
    x=\overbrace{1\cdots 1}^{k}\overbrace{ 1\cdots 1}^{k}\overbrace{ 0\cdots 0}^{k}\overbrace{ 0\cdots 0}^{k+1}1,$ and $y=\overbrace{1\cdots 1}^{k}\overbrace{ 0\cdots 0}^{k}\overbrace{ 1\cdots 1}^{k}\overbrace{ 0\cdots 0}^{k+1}1.
$
 We have $H(x,y')=k+k+1=\frac{n}{2}$ and ${\cal P}(x,y')=0$. Since $(x,x')\in R$ and $(y,y')\in R$, we have $(x,y')\in R$ and  ${\cal P}(x,y')=0$, which is a contradiction.

 According to Corollary 1.2 from \cite{Fr87}, we have $|S|\leq 1.99^n$. Therefore, the minimum number of 1-monochromatic rectangles that partition the space of inputs is
 \begin{equation}
     C^1(\mbox{EQ}')\geq\frac{|E|}{|S|}\geq \frac{{n-1\choose  n/2-1}}{(1.99)^{n-1}}>\frac{2^{n-1}/(n-1)}{(1.99)^{n-1}}.
\end{equation}
The deterministic communication complexity is then
 \begin{equation}
D(\mbox{EQ}')\geq \log{C^1(\mbox{EQ}')}>\log{\frac{2^{n-1}/(n-1)}{(1.99)^{n-1}}}>0.0073n.
 \end{equation}

Therefore the theorem has been proved.
\end{proof}

\section{State succinctness results}
Now we are ready to derive the state succinctness result.

\begin{theorem}
 The promise problem $A(n)$ can be solved exactly by  a 1QCFA ${\cal A}(n)$  with ${\bf O}(n^2)$ quantum basis states and ${\bf O}(n^3)$ classical states, whereas the sizes of the corresponding 1DFA  are $2^{{\bf \Omega}(n)}$.
\end{theorem}

\begin{proof}
Let $x=x_1\cdots x_n$ and $y=y_1\cdots y_n$ be in $\{0,1\}^n$. The input word on the tape will be $w=\cent x\#y\#\#x\#y\$$.  Let us consider a 1QCFA ${\cal A}(n)$ with ${\bf O}(n^2)$   quantum basis states
$\{|0,0\rangle,|i,0\rangle,|i,j\rangle,|k\rangle:0\leq i<j\leq n,\ 1\leq k\leq n-2  \}$ and ${\bf O}(n^3)$  classical states $\{S_{ijp}: 0\leq i,j, p\leq n+1\}$ (some of the states may be not used in the automaton actions).

 ${\cal A}(n)$  starts in  the initial quantum state  $|0,0\rangle$ and the initial classical state $S_{000}$. We use classical states $S_{ijp}$ ($1\leq p\leq n+1$) to point out the positions of the tape head that will provide some information for  quantum transformations. If  the classical state of ${\cal A}(n)$ will be $S_{ijp}$ ($1\leq p\leq n$) that will  mean that the next scanned symbol of the tape head is the $p$-th symbol of $x$($y$) and $S_{ijn+1}$ means that the next scanned symbol of  the tape head  is  $\#$($\$$).

  The behavior of ${\cal A}(n)$ is composed of two parts. The first part is the behavior of ${\cal A}(n)$ when reading the prefix of the input, namely $\cent x\#y\#$. In this part,  ${\cal A}(n)$ uses quantum basis state $\{|0,0\rangle,|i,0\rangle,|i,j\rangle:0\leq i<j\leq n  \}$ and  classical states $S_{00p}$ ($0\leq p\leq n+1$) to simulate Steps 1 to  6 in the proof of Theorem \ref{Th-1}.
After the measurement   at the end of the first part, if the outcome is $|0,0\rangle$, then the input is accepted. Otherwise, suppose  that the outcome is $|i,j\rangle$, the classical state will be changed to $S_{ij0}$ ($1\leq i<j\leq n$, which  means that $H(x_ix_j,y_iy_j)=1$ and the input bits $x_i,x_j,y_i,y_j$ will be skipped during the second part of the behavior of ${\cal A}(n)$.
The second part is the behavior of the automation when reading the second part of the input $\#x\#y\$$.  In this part,   ${\cal A}(n)$  uses quantum basis states
$\{|k\rangle:\ 1\leq k\leq n-2  \}$  and classical states $S_{ijp}$ ($0\leq p\leq n+1$) to simulate Steps 7 to 12 in the proof of Theorem \ref{Th-1}.
 The automaton proceeds as follows:

 \begin{enumerate}
\item ${\cal A}(n)$  reads the left end-marker $\cent $,  performs $U_1$ on the initial quantum state $|0,0\rangle$,  changes its classical state to $\delta(S_{000},\ \cent  )=S_{001}$, and moves the tape head one cell to the right.

\item Until the currently  scanned symbol $\sigma$ is not $\#$, ${\cal A}(n)$ does the following:
 \begin{enumerate}
 \item Applies $\Theta(S_{00p},\sigma)=U_{p,\sigma}$ to the current quantum state.
 \item Changes the classical state $S_{00p}$ to $S_{00p+1}$ and moves the tape head one cell to the right.
\end{enumerate}
\item ${\cal A}(n)$ changes the classical state $S_{00p+1}$ to  $S_{001}$ and moves the tape head one cell to the right.

\item Until the currently  scanned symbol $\sigma$ is not $\#$, ${\cal A}(n)$ does the following:
 \begin{enumerate}
 \item Applies $\Theta(S_{00p},\sigma)=U_{p,\sigma}$ to the current quantum state.
 \item Changes the classical state $S_{00p}$ to $S_{00p+1}$ and moves the tape head one cell to the right.
\end{enumerate}

\item When  $\#$  is reached,   ${\cal A}(n)$ performs $U_{2}$ on the current quantum state.
\item
${\cal A}(n)$ measures the current quantum state in the standard basis.   If the outcome is $|0,0\rangle$, ${\cal A}(n)$ accepts the input; otherwise,  suppose that the outcome is $|i,j\rangle$, ${\cal A}(n)$ changes the classical state to $S_{ij0}$, moves the tape head one cell to the right.

\item ${\cal A}(n)$ reads  $\#$,  applies $\Theta(S_{ij0},\#)=U_3U_{ij}$ to the current quantum state,  changes its classical state to $S_{ij1}$, and moves the tape head one cell to the right.

\item Until the currently  scanned symbol $\sigma$ is not $\#$, ${\cal A}(n)$ does the following:
 \begin{enumerate}
 \item Applies $\Theta(S_{ijp},\sigma)=U_{ijp,\sigma}$ to the current quantum state.
 \item Changes the classical state $S_{ijp}$ to $S_{ijp+1}$ and moves the tape head one cell to the right.
\end{enumerate}
\item ${\cal A}(n)$ changes the classical state $S_{ijp+1}$ to  $S_{ij1}$ and moves the tape head one cell to the right.

\item While the currently  scanned symbol $\sigma$ is not the right end-marker $\$$, ${\cal A}(n)$ does the following:
 \begin{enumerate}
 \item Applies $\Theta(S_{ijp},\sigma)=U_{ijp,\sigma}$ to the current quantum state.
 \item Changes the classical state $S_{ijp}$ to $S_{ijp+1}$ and moves the tape head one cell to the right.
\end{enumerate}

\item When the right end-marker  is reached,  ${\cal A}(n)$  performs $U_{4}$ on the current quantum state.
\item
${\cal A}(n)$ measures the current quantum state in the standard basis.   If the outcome is $|1\rangle$, ${\cal A}(n)$ accepts the input; otherwise,  rejects the input.
\end{enumerate}
where unitary transformations  $U_1$, $U_2$, $U_3$, $U_4$  are the  ones defined in the proof of Theorem \ref{Th-1} and
\begin{eqnarray*}
U_{p,\sigma}|i,j\rangle&=&(-1)^{\sigma}|i,j\rangle  \    \mbox{if}\ i=p;\\
U_{p,\sigma}|i,j\rangle&=& |i,j\rangle \ \mbox{if}\ i\neq p;\\
U_{ij}|i,j\rangle&=& |1\rangle;\\
U_{ijp,\sigma}|k\rangle&=&(-1)^{\sigma}|k\rangle  \    \mbox{if}\ p<i\ \mbox{and}\ k=p; \\
U_{ijp,\sigma}|k\rangle&=&|k\rangle  \    \mbox{if}\ p<i\ \mbox{and}\ k\neq p; \\
U_{ijp,\sigma}|k\rangle&=&|k\rangle  \    \mbox{if}\ p=i; \\
U_{ijp,\sigma}|k\rangle&=&(-1)^{\sigma}|k\rangle  \    \mbox{if}\ i<p<j\ \mbox{and}\ k=p-1; \\
U_{ijp,\sigma}|k\rangle&=&|k\rangle  \    \mbox{if}\ i<p<j\ \mbox{and}\ k\neq p-1; \\
U_{ijp,\sigma}|k\rangle&=&|k\rangle  \    \mbox{if}\ p=j; \\
U_{ijp,\sigma}|k\rangle&=&(-1)^{\sigma}|k\rangle  \    \mbox{if}\ p>j\ \mbox{and}\ k= p-2; \\
U_{ijp,\sigma}|k\rangle&=&|k\rangle  \    \mbox{if}\ p>j\ \mbox{and}\ k\neq p-2.
\end{eqnarray*}

It is easy to verify that unitary transformation $U_{p,\sigma}$,  $U_{ij}$ and $U_{ijp,\sigma}$ exist.
The rest  of the proof is   analogues to the proof in   Theorem \ref{Th-1}.

According to Theorem \ref{Th-lowerbound} , $D(\mbox{EQ}')\in {\bf \Omega}( n)$. Therefore, it is easy to see that the sizes of the corresponding 1DFA for  $A(n)$ are $2^{{\bf \Omega}(n)}$ \cite{KusNis97b}.
\end{proof}

\section*{Acknowledgements}
The first author would  like to thank Jozef Gruska for having a most constructive influence on him, teaching him how to think and write in a clear way, and for continuous support during his  postdoctoral research  with him.

\end{document}